\documentclass[12pt]{article}
\usepackage{tikz}
\usepackage{graphicx}
\usepackage{verbatim}
\usepackage{fullpage}

\usepackage{url}
\usepackage{lipsum}
\usepackage{enumitem}

\usepackage{amssymb,amsfonts,amsmath,amsthm}
\usepackage[colorlinks=true,linkcolor=blue,citecolor=red]{hyperref}
\usepackage[ruled,vlined,lined,commentsnumbered,linesnumbered]{algorithm2e}
\usepackage{algpseudocode}
\usepackage{epstopdf}

\newtheorem{theorem}{Theorem}[section]
\newtheorem{corollary}[theorem]{Corollary}

\newtheorem{lemma}[theorem]{Lemma}

\title{New Results On Routing Via Matchings On Graphs}
\author{Avah Banerjee, Dana Richards}
\newcommand{\FormatAuthor}[2]{
\begin{tabular}{c}
#1 \\ {\small #2}
\end{tabular}
}
\author{
\begin{tabular}[h!]{lcr}
   \FormatAuthor{Avah Banerjee}{George Mason University}
&   \FormatAuthor{Dana Richards}{George Mason University}
\end{tabular}
}

\begin{document}

\maketitle

\begin{abstract}
In this paper we present some new complexity results on  the routing time of a graph under the \textit{routing via matching} model. 
This is a parallel routing model which was introduced by Alon et al\cite{alon1994routing}. 
The model can be viewed as a communication scheme on a distributed network. 
The nodes in the network can communicate via matchings (a step), where a node exchanges data (pebbles) with its matched partner.
Let $G$ be a connected graph with vertices labeled from $\{1,...,n\}$ and the destination vertices of the pebbles are given by a permutation $\pi$.
The problem is to find a minimum step routing scheme for the input permutation $\pi$.
This is denoted as the routing time $rt(G,\pi)$ of $G$ given $\pi$. 
In this paper we characterize the complexity of some known problems under the routing via matching model and discuss their relationship to graph connectivity and clique number.
We also introduce some new problems in this domain, which may be of
independent interest.

\end{abstract}

%%%%%%%%%%%%%%%%%%%%%%%%%%%%%%%%%%%%%%%%%%%%%%%%%
\section{Introduction}
%%%%%%%%%%%%%%%%%%%%%%%%%%%%%%%%%%%%%%%%%%%%%%%%%
Originally introduced by Alon and others \cite{alon1994routing} the routing via matching model explores a parallel routing problem on connected undirected graphs.
Consider a undirected labeled graph $G$.
Each vertex of $G$ contains a pebble with the same label as the vertex.
Pebbles move along edges through a sequence of swaps.
A set of swaps (necessarily disjoint) that occurs concurrently is called a step.
This is determined by a matching.
A permutation $\pi$ gives the destination of each pebble. 
That is, the pebble $p_v$ on vertex $v$  is destined for the vertex $\pi(v)$.
The task is to route each pebble to their destination via a sequence of matchings. 
 The \textit{routing time} $rt(G,\pi)$ is defined as the minimum number of steps necessary to route all the pebbles for a given permutation.
 The \textit{routing number} of $G$, $rt(G)$, is defined as the maximum routing time over all permutation. Let $G=(V,E)$, $m=|E|$ and $|G|=n=|V|$.
 
 Determining the routing time is a special case of the \textit{minimum generator sequence} problem for groups.
 In this problem instead of a graph we are given a permutation group $\mathcal{G}$ and a set of generators $S$.
 Given a permutation $\pi \in \mathcal{G}$ the task is to determine if there exists a generator sequence of length $\le k$ that generates $\pi$  from the identity permutation.
 It was first shown to be $ \mathsf{NP} $-hard by Evan and Goldreich \cite{even1981minimum}.
 Later Jerrum \cite{jerrum1985complexity} showed that it is in fact $\mathsf{PSPACE}$-complete, even when the generating set is restricted to only two generators.
 
 The serial version,
 where swaps takes place one at a time, is also of interest.
 This has recently garnered interest after its introduction 
 by Yamanaka and others \cite{yamanaka2015swapping}.
 They have termed it the \textit{token swapping problem}. 
 This problem is also $\mathsf{NP} $-complete as shown by Miltzow and others \cite{miltzow2016approximation} in a recent paper.
  Where the authors prove token swapping problem is hard to approximate within $(1+ \delta)$ factor. 
  They also provide a simple 4-approximation scheme for the problem.
  A generalization of the token swapping problem (and also the permutation routing problem) is the colored token swapping problem \cite{yamanaka2015swapping,kawahara2016time}. 
 In this model the vertices and the tokens are partitioned into equivalence classes (using colors) and the goal is to route all  pebbles in such a way that each pebble ends up in some vertex with the same class as the pebble. 
 If each pebble (and vertex) belong to a unique class then this problem reduces to the original token swapping problem.
 This problem is also proven to be $\mathsf{NP} $-complete by Yamanaka and others \cite{yamanaka2015swapping} when the number of colors is at least 3.
 The problem is polynomial time solvable for the bi-color case.
%%%%%%%%%%%%%%%%%%%%%%%%%%%%%%%%%%%%%%%%%%%%%%%%%
\subsection{Prior Results}
%%%%%%%%%%%%%%%%%%%%%%%%%%%%%%%%%%%%%%%%%%%%%%%%%
  Almost all previous literature on this problem focused on determining the routing number of typical graphs.  
  In the introductory paper, Alon and others \cite{alon1994routing} show that for any connected graph $G$, $rt(G) \le 3n$. 
  This was shown by considering a spanning tree of $G$ and using only the edges of the tree to route  permutations in $G$. 
  Note that, one can always route a permutation on a tree, by iteratively moving a pebble that belongs to some leaf node and ignoring the node afterward. 
  The routing scheme is recursive and uses an well known property of trees: a tree has a centroid (vertex) whose removal results in a forest of trees with size at most $n/2$.
Later Zhang and others \cite{zhang1999optimal} improve this upper bound to $3n/2 + O(\log n)$.
 This was done using a new decomposition called the caterpillar decomposition. 
 This bound is essentially tight as it takes $\lfloor{3(n-1)/2}\rfloor$ steps to route some permutations on a star $K_{1,n-1}$. 
 There are also some known results for routing numbers of graphs besides trees. 
 We know that for the complete graph and the complete bipartite graph the routing number is 2 and 4 respectively \cite{alon1994routing}.
 Where the latter result is attributed to W. Goddard.
  Li and others \cite{li2010routing} extend these results to show $rt(K_{s,t}) = \lfloor 3s/2t \rfloor + O(1)$ ($s \ge t$). 
 For the $n$-cube $Q_n$ we know that $n+1 \le rt(Q_n) \le 2n-2$. The lower bound is quite straightforward. The upper bound was discovered by determining the routing number of the Cartesian product of two graphs \cite{alon1994routing}.
  If $G = G_1  \square  G_2$ be the Cartesian product of $G_1$ and $G_2$ then:
   $$rt(G) \le 2 \min(rt(G_1),rt(G_2))+\max(rt(G_1),rt(G_2))$$
Since $Q_n = K_2 \square Q_{n-1}$ the result follow.\footnote{The base case, which computes $rt(Q_3)$ was determined to be 4 via a computer search\cite{li2010routing}}.

\subsection{Our Results}
In this paper we present several complexity results for the routing time problem and some variants of it. We summarize these results below.

\begin{enumerate}[series = inform, leftmargin =\parindent, labelwidth =0pt, listparindent = \parindent]
\item[] Complexity results on routing time:
\end{enumerate}

\begin{enumerate}[labelindent =\parindent, align = left, labelwidth =1em, leftmargin =! , labelsep =0.2 em, series = quest]

    \item If $G$ is at least bi-connected then determining whether $rt(G, \pi) = k$ for any arbitrary permutation and $k > 2$ is NP-complete.
    \footnote{After publication of our results to arXiv (\cite{banerjee2016routing}) a similar result was independently discovered in the context of parallel token swapping by Kawahara and others \cite{kawahara2017time}.}
    
    \item For any graph, determining $rt(G, \pi) \le 2$ can be done in polynomial time, for which we give a $O(n^{2.5})$ time algorithm.
    
   \item As a consequence of our $\mathsf{NP}$-completeness proof of the routing time we show that the problem of determining a minimum sized partitioning scheme of a colored graph such that each partition induces a connected subgraph is $\mathsf{NP}$-complete.
	
\item We introduce a notion of approximate routing called \textit{maximum routability} for a graph and give an approximation algorithm for it.

\end{enumerate}

\begin{enumerate}[series = inform, leftmargin =\parindent, labelwidth =0pt, listparindent = \parindent]
\item[] Structural results on routing number:
\end{enumerate}

\begin{enumerate}[resume* = quest]
	\item If $G$ is $h$-connected then $G$ has a routing number  of $O(nr_G)$. Here $r_G =  \min rt(G_h)/|G_h|$, over all induced connected subgraphs $|G_h| \le h$.
	\item A connected graph with a clique number of $\kappa$ has a routing number of $O(n -\kappa)$.
\end{enumerate}

Routing on general graphs is a natural question and the swapping model is a natural model in synchronous networks.
Our results are some of the first to address these models when 
the graph has certain topological properties.  
Connectivity properties are basic, especially for network algorithms.
While the hope is to have matching upper and lower bounds for, say, $h$-connected graphs,
we give new algorithms and techniques towards that end.

%%%%%%%%%%%%%%%%%%%%%%%%%%%%%%%%%%%%%%%%%%%%%%%%%
\section{Computational Results}
%%%%%%%%%%%%%%%%%%%%%%%%%%%%%%%%%%%%%%%%%%%%%%%%%
%%%%%%%%%%%%%%%%%%%%%%%%%%%%%%%%%%%%%%%%%%%%%%%%%
\subsection{An $O(n^{2.5})$ time Algorithm for Determining If $rt(G,\pi) \le 2$}
%%%%%%%%%%%%%%%%%%%%%%%%%%%%%%%%%%%%%%%%%%%%%%%%%
In this section we present a polynomial time deterministic algorithm to compute a two step routing scheme if one exists.
 It is trivial to determine whether $rt(G,\pi) = 1$.
 Hence, we only consider the case if $rt(G,\pi) > 1$.
 The basic idea centers around whether we can route the individual cycles of the permutation within 2 steps. 
 Let $\pi = \pi_1\pi_2\ldots \pi_k$ consists of $k$ cycles and $\pi_i = (\pi_{i,1}\ldots \pi_{i,a_i})$, where $a_i$ is the number of elements in $\pi_i$.
 A cycle $\pi_i$ is identified with the vertex set $V_i \subset V$ whose pebbles need to be routed around that cycle.
  We say a cycle $\pi_i$ is  \textit{self-routable} if it can be routed on the induced subgraph $G[V_i]$ in 2 steps. 
  
If all cycles were self-routable we would be done, so suppose that there is a cycle $\pi_i$ that 
needs to match across an edge between it and another cycle $\pi_j$.
Let $G[V_i,V_j]$ be the induced bipartite subgraph corresponding to the two sets $V_i$ and $V_j$.

\begin{figure}[h]
	\includegraphics[width=3.5cm]{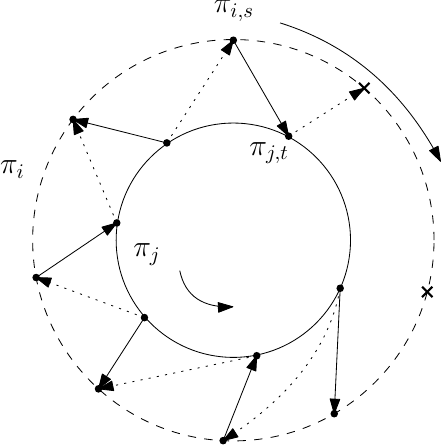}
	\centering
	\caption{The two cycles are shown as concentric circles. The direction of rotation for the outer circle is clockwise and the inner circle is counter-clockwise. Once, we choose $(\pi_{i,s},\pi_{j,t})$ as the first matched pair, the rest of the matching is forced. Solid arrows indicate matched vertices during the first round. Note that the cycles are unequal and the crossed vertices in the figure will not be routed.} 
\end{figure}

\begin{lemma}
	If $\pi_i$  is not self-routable and it is routed with an edge from $V_i$ to $V_j$ then $\pi_i$ and $\pi_j$ are both routable in 2 steps when all of the edges used are from $G[V_i,V_j]$ and
when $|V_i|=|V_j|$.
\end{lemma}

\begin{proof}
	We prove this assuming $G$ is a complete graph. Since for any other case the induced subgraph $G[V_i\cup V_j]$ would have fewer edges, hence this is a stronger claim. 
Let the cycle $\pi_i=(\pi_{i,1},\ldots,\pi_{i,s},\ldots,\pi_{i,|V_i|})$. 
If there is an edge used between the cycles then there must be such an edge in the first step,
since pebbles need to cross from one cycle to another and back. 
Assume $\pi_{i,s}$ is matched with $\pi_{j,t}$ in the first step. From Figure 1 we see that the crossing pattern is forced, and unless $|V_i|=|V_j|$, the pattern will fail.
\end{proof}

\noindent
A pair of cycles $\pi_i,\pi_j$ is \textit{mutually-routable} in the case described by Lemma 1.
Naively verifying whether a cycle $\pi_i$ is self-routable, or a pair $(\pi_i,\pi_j)$ is 
mutually-routable takes $O(|V_i|^2)$ and $O((|V_i| + |V_j|)^2)$ time respectively. 
However, with additional bookkeeping we can compute this in linear time on the size of the induced graphs. This can be done by considering the fact that no edge can belong to more than one routing scheme on $G[V_i]$ or on $G[V_i,V_j]$. Hence the set of edges are partitioned by the collection of 2 step routing schemes. Self-routable schemes, if they exist, are forced by the choice of any edge to be in the first step; no edge is forced by more than four initial choices, leading to a test that runs in time proportional in $|G[V_i]|$.  
Mutually-routable schemes, if they exist, are one of $|V_i|$ ($=|V_j|$) possible schemes; each edge votes for a scheme and it is routable if a scheme gets enough votes, leading to a test that runs in time proportional in $|G[V_i,V_j]|$.
All the tests can be done in $O(m)$ time.

We define a graph $G_{cycle} = (V_{cycle} ,E_{cycle})$ whose vertices are the cycles ($V_{cycle} = \{\pi_i\}$) and two cycle are adjacent iff they are mutually-routable in 2 steps. 
Additionally, $G_{cycle}$ has loops corresponding to vertices which are self-routable cycles. 
We can modify any existing maximum matching algorithm to check whether $G_{cycle}$ has a perfect matching (assuming self loops) with only a linear overhead. 
We omit the details. 
Then the next lemma follows immediately:
\begin{lemma}
	$rt(G,\pi) = 2$ iff there is a perfect matching in $G_{cycle}$.
\end{lemma}
\noindent The graph $G_{cycle}$ can be constructed in $O(m)$ time by determining self and mutual routability of cycles and pair of cycles respectively. 
Since we have at most $k$ cycles,  $G_{cycle}$ has $\le 2k$ vertices and thus $O(k^2)$ edges. 
Hence we can determine a maximum matching in $G_{cycle}$ in  $O(k^{2.5})$ time \cite{micali1980v}. 
This gives a total runtime of $O(n + m + k^{2.5})$ for our algorithm to find a 2-step routing scheme of a connected graph if one exists.
\begin{corollary}
	$rt(G) = 2$ iff $G$ is a clique.
\end{corollary}
\begin{proof}
($\Rightarrow$ ) A two step routing scheme for $K_n$ was given in \cite{alon1994routing}.\\
($\Leftarrow $) If $G$ is not a clique then there is at least a pair of non-adjacent vertices. Let $(i,j)$ be a non-edge. Then by Lemma 1 the permutation $(ij)(1)(2)\ldots(n)$ cannot be routed in two steps.
\end{proof}

%%%%%%%%%%%%%%%%%%%%%%%%%%%%%%%%%%%%%%%%%%%%%%%%%
\subsection{Determining $rt(G,\pi) \le k$ Is Hard for Any $k \ge 3$}

\begin{theorem}
	For $k \ge 3$ computing $rt(G,\pi)$ is $\mathsf{NP}$-complete.
\end{theorem}
\begin{proof}
Proving it is in $ \mathsf{NP} $ is trivial, we can use a set of matchings as a witness. We give a reduction from 3-$\mathsf{SAT}$.
We first define three \textit{atomic} gadgets (see Figure 2) which will be use to construct the variable and clause gadgets. Vertices whose pebbles are fixed (1-cycles) are represented as red circles.  Otherwise they are represented as black dots. So in the first three sub-figures ((a)-(c)) the input permutation is $(a,b)$\footnote{We do not write the 1 cycles explicitly as is common.}. In all our constructions we shall use permutations consisting of only 1 or 2 cycles. Each cycle labeled $i$ will be represented as the pair $(a_i,b_i)$. If the correspondence between a pair is clear from the figure then we shall omit the subscript.  It is an easy observation that $rt(P_{3},((a,b))) = rt(P_{4},((a,b))) = rt(H,((a,b))) = 3$. In the case of the hexagon $H$ we see that in order to route the pebbles within 3 steps we have to use the left or the right path, but we cannot use both paths simultaneously (i.e., $a$ goes along the left path but $b$ goes along the right and vice-versa). Figure 2(e) shows a chain of diamonds connecting $u$ to $v$. Where each diamond has a 2-cycle, top and bottom. If vertex $u$ is used to route any pebble other than the two pebbles to its right then the chain construction forces $v$ to be used in routing the two pebbles to its left. This chain is called a \textit{diamond-chain}. In our construction we only use chains of constant length to simplify the presentation of our construction. 

\begin{figure}[h]
	\includegraphics[width=11cm]{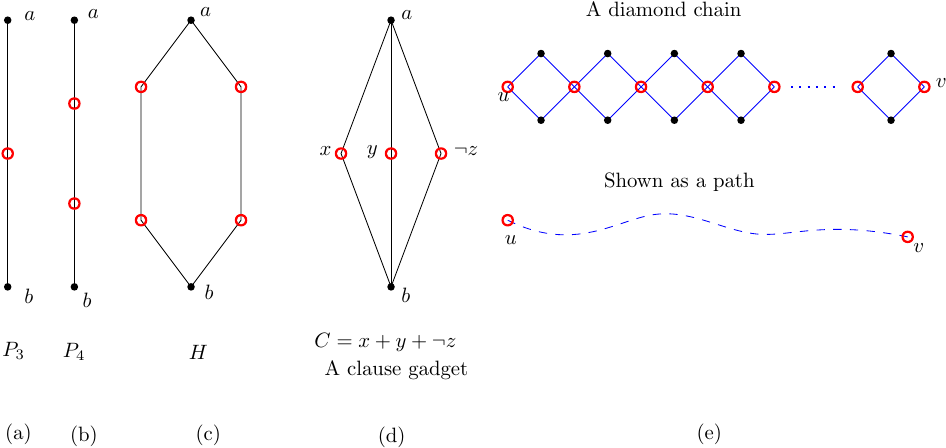}
	\centering
	\caption{Atomic Gadgets, pairs $ (a,b) $ need to swap their pebbles. The unmarked red circles have pebbles that are fixed.} 
\end{figure}

	\begin{figure}[h]
		\includegraphics[width=11cm]{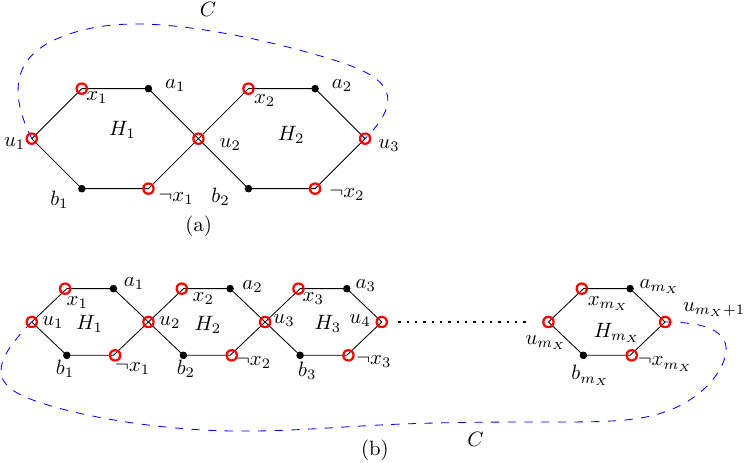}
		\centering
		\caption{Variable graph of $X$. (a) is a special case for $m_X=2$, (b) is the general case.} 
	\end{figure}
	\subsubsection{Clause Gadget:} Say we have clause $C = x \vee y \vee \neg z$. In Figure 2(d) we show how to create a clause gadget. This is referred to as the \textit{clause graph} $G_C$ for the clause $C$. The graph in Figure 2(d) can route $\pi_C=(a_C,b_C)$ in three steps by using one of the three paths between $a_C$ and $b_C$. Say, $a_C$ is routed to $b_C$ via $x$. Then it must be the case that vertex $x$ is not used to route any other pebbles. We say the vertex $x$ is \textit{owned} by the clause. Otherwise, it would be not possible to route $a_C$ to $b_C$ in three steps via $x$. We can interpret this as follows. A clause has a satisfying assignment iff its clause graph has a owned vertex.
	
	\subsubsection{Variable Gadget:} Construction of the variable gadgets is done in a similar manner. 
	The variable gadget $G_X$ corresponding to $X$ is shown in Figure 3(b).
	Figure 3(a) is essentially a smaller version of 4(b) and is easier to understand.
	If we choose to route $a_1$ and $b_1$ via the top-left path passing through $x_1$ and $u_1$ then $(a_2, b_2)$ must be routed via $x_2$ and $u_2$.
	This follows from the fact that since $u_1$ is occupied the pebbles in the diamond chain $C$ (the dashed line connecting $u_1$ with $u_3$) must use $u_3$ to route the right most pair.
	By symmetry, if we choose to route $(a_1, b_1)$ using the bottom right path (via $\neg x_1$, $u_2$) then we also have to choose the bottom right path for $(a_2,b_2)$.
	These two (and only two) possible (optimal) routing scheme can be interpreted as variable assignment.
	Let $G_X$ be the graph corresponding to the variable $X$ (Figure 3(b)).
	The top-left routing scheme leaves the vertices $\neg x_1, \neg x_2, \ldots $ free to be used for other purposes since they will not be able take part in routing pebbles in $G_X$.
	Thus this can be interpreted as setting the variable $X$ to false.
	This ``free" vertex can be used by a clause (if the clause has that literal) to route its own pebble pair.
	That is they can become owned vertices of some clause.
	Similarly, the bottom right routing scheme can be interpreted as setting $X$ to true.
	For each variable we shall have a separate graph and a corresponding permutation on its vertices.
	The permutation we will route on $G_X$ is $\pi_X = (a_1b_1)(a_2b_2)\ldots(a_{m_X},b_{m_X})\pi_{f_X}$.
	The permutation $\pi_{f_X}$ corresponds to the diamond chain connecting $u_1$ with $u_{m_X+1}$.
	The size of the graph $G_X$ is determined by $m_X$, the number of clauses the variable $X$ appears in.
	
	\subsubsection{Reduction:} For each clause $C$, if the literal $x \in C$  then we connect $x_i \in G_X$  (for some $i$) to the vertex labeled $x \in G_C$ via a diamond chain. If $\neg x \in C$ then we connect it with $\neg x_i$ via a diamond chain. This is our final graph $G_{\phi}$ corresponding to an instance of a 3-$ \mathsf{SAT} $ formula. The input permutation is $\pi = \pi_X\ldots\pi_C\ldots\pi_f\ldots$, which is the concatenation of all the individual permutations on the variable graphs, clause  graphs and the diamond chains.  This completes our construction. We need to show, $rt(G_\phi,\pi) = 3$ iff $\phi$ is satisfiable. Suppose $\phi$ is satisfiable. Then for each variable $X$, if the literal $x$ is true then we use bottom-right routing on $G_X$, otherwise we use top-left routing. This ensures in each clause graph there will be at least one owned vertex.  Now suppose $(G_{\phi},\pi) = 3$. Then each clause graph has at least one owned vertex. If $x$ is a free vertex in some clause graph then $\neg x$ is not a free vertex in any of the other clause graphs, otherwise variable graph $G_X$ will not be able route its own permutation in 3 steps. Hence the set of free vertices will be a satisfying assignment for $\phi$. It is an easy observation that the number of vertices in $G_\phi$ is polynomially bounded in $n,m$; the number of variables and clauses in $\phi$ respectively and that $G_\phi$ can be explicitly constructed in polynomial time.	
\end{proof}

\begin{corollary}
Computing $rt(G, \pi)$ remains hard even when $G$ is restricted to being 2-connected.
\end{corollary}

\subsection{Connected Colored Partition Problem ($ \mathsf{CCPP} $)}
Our proof technique for Theorem 1 can be used to prove that the following problem is also $\mathsf{NP} $-hard.
Let $G$ be a graph whose vertices are colored  with $k$ colors. 
We say a partition $\mathcal{S}=\{S_1,\ldots, S_r\}$ of the the vertex set $V$ respects the coloring $C$ (where $C : V \to \{1,\ldots,k\}$) if each partition either contains all vertices of some color or none of the vertices of that color (necessarily $r \le k$). Further, we require the induced subgraph $G[S_i]$ be connected, for every $i$.
Given a graph $G$, a coloring $C$ (with $k$ colors) and a integer $t \le n$ the decision version of the problem asks, whether there exists a valid partitioning whose largest block has a size of at most $t$.
We denote this problem by $\mathsf{CCPP}(G, k, t)$. 
If we replace the requirement of connectedness of the induced subgraphs with other efficiently verifiable properties then it is a strict generalization of the better known monochromatic partitioning problems on colored graphs (see for example \cite{gyarfas2011partitioning}).
Note that the connectivity requirement on the induced subgraphs is what makes this problem graphical.
In fact without it the problem becomes trivial, as one can simply partition the vertices into monochromatic sets, which is the best possible outcome.
$\mathsf{CCPP}(G, k, t)$ is in $ \mathsf{P} $ if $t$ is constant. Since one can simply enumerate all partitions and there are $O(k^t)$ of them.

\begin{theorem}
	$\mathsf{CCPP}(G, k, t)$ is $ \mathsf{NP} $-complete for arbitrary $k$ and $t$. 
\end{theorem}

\begin{figure}[h]
	\includegraphics[width=3cm]{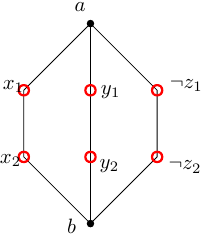}
	\centering
	\caption{A modified clause gadget (from the proof of Theorem 1) for the clause $C = x + y + \neg z$.} 
\end{figure}

\begin{proof}
	The proof essentially uses a similar set of gadgets as used in the proof of Theorem 1. The idea is to interpret a route as a connected partition. 
	We show that even when each color class is restricted to at most two vertices the problem remain $\mathsf{NP} $-complete. This done via a reduction from the 3-$\mathsf{SAT}$ problem. We reuse some of the gadgets from the proof of Theorem 1, but we interpret them differently. Lets discuss the clause gadgets first. In Figure 4 we see the graph corresponding to the clause $C = x + y + \neg z$. The vertices $a$ and $b$ have the same color, lets say $c$ so that they identify with the clause $C$. All the other vertices (shown as red circles) have unique colors which also differs from $c$. Clearly any valid partition of $C$ must include both pair of vertices along some line. Hence in an optimal partition of the clause graph the largest subgraph is of size 4.
	
	The variable gadgets are  same as before (see Figure 3) except we do away with the diamond chain and fuse the two end vertices together. So for example in Figure 3 (a), $u_1 = u_{3}$. We also make the gadgets twice as long, so instead of having ${m_X}$ hexagons we now have $2m_X$ number of them. The vertices $(a_i,b_i)$ have the same color but which is distinct from every other vertices. As with the clause graphs the red vertices in the variable gadgets all have unique colors.
	
	To construct the graph $G_\phi$ corresponding to a boolean formula $\phi$ we do the following. Since we are not using the diamond chains anymore we directly fuse vertices from clause graphs and variable graphs. So for example if the clause $C$ has variable $x$ as a true literal we fuse  two vertices $x_i$ and $x_j$ (for some $i,j \le 2m_X$) to the two vertices $x_1$ and $x_2$ in the clause graph. That is $x_1, x_2 \in \{x_i,x_j\}$.
	In every partition the pair $a_i, b_i$ must be included since they have the same color. Since each partition must be connected this can happen if either we take the segment $(a_i,x_i,u_i,b_i)$ or the segment $(a_i,u_{i+1 \mod 2m_X}, \neg x_{i+1 \mod 2m_x},b_i)$. Lets call them top-left and bottom-right segments.
	Clearly, if we take the top-left segment as part of some partition for any pair $(a_i,b_i)$ we have to use the corresponding top-left segments for all other $(a,b)$ pairs in the same variable graph. 
	Otherwise, the segments will not be connected.
	Same is true with the bottom-right segments. 
	This forces variable assignments. 
	If we choose the top left segments then the bottom-right vertices ($\neg x_i$'s) which corresponds to the negated literals will be free and a pair of them can be use to partition a clause graph which contains those literal vertices.
	If $G_{\phi}$ has a partitioning scheme such that every partition is of size 4 then $\phi$ is satisfiable. We can look at the partition to determine which segments were chosen from the variable graph which determines the variable assignment. Since every clause graph has been partitioned into components of size at most 4, we conclude that every clause is satisfied.
	The other direction can be proven in a similar manner. 
\end{proof}

%%%%%%%%%%%%%%%%%%%%%%%%%%%%%%%%%%%%%%%%%%%%%%%%%
\subsection{Routing As Best You Can}
%%%%%%%%%%%%%%%%%%%%%%%%%%%%%%%%%%%%%%%%%%%%%%%%%
It is often desirable to determine how many packets we can send to their destination within a certain number of steps.
Such as propagating information in social media.
In the context of permutation routing this leads to a notion of \textit{maximum routability}. 
Given two permutation $\pi$ and $\sigma$ let $|\pi - \sigma|$ be the number of fixed points in $\tau$ such that $\tau \pi = \sigma$.
Let us define \textit{maximum routability} $mr(G,\pi,k)$ as follows:
$$ mr(G,\pi,k) = \max_{\sigma \in S_n ,\ rt(G, \sigma) \le k}{|\pi-\sigma|}$$
We denote by $ \mathsf{MaxRoute} $ the problem of computing maximum routability.
\noindent Essentially, $\sigma$ is a permutation out of all permutations that can be routed in $\le k$ steps and that has the maximum number of elements in their correct position as given by $\pi$.
 The permutation $\sigma$ may not be unique.
  It can be easily shown (as a corollary to Theorem 1) that the decision version of this problem is $\mathsf{NP}$-hard, since we can determine $rt(G, \pi)$ by asking whether $mr(G,\pi, k) = n$. 
(Of course $rt(G, \pi) = O(n)$ for any graph, hence $O(\log n)$ number of different choices of $k$ is sufficient to compute  $rt(G, \pi)$ exactly.)

In this section we give an approximation algorithm for computing the maximum routability when the input graph $G$ satisfies the following restriction. 
If the maximum degree of $G$ is $\Delta$ such that $(\Delta + 1) ^ k = O(\log^2 n) $ then $mr(G, \pi, k)$ can be approximated within a factor of $O(n\log\log n/ \log n)$ from the optimal.  
Unfortunately a good approximation for $rt(G, \pi)$ does not lead to a good approximation ratio when computing $mr(G, \pi, k)$ for any $k > 2$. 
The reason being that in an optimal algorithm for routing $\pi$ on $G$ it is conceivable that all pebbles are displaced at the penultimate stage and the last matching fixes all the displaced pebbles. 

Our approximation algorithm is based on a reduction to the $\mathsf{MaxClique}$ problem.
The $ \mathsf{MaxClique} $ problem has been extensively studied. 
In fact it is one of the defining problems for $\mathsf{PCP}$-type systems of probabilistic verifiers \cite{feige1996interactive}. 
It has been shown that $ \mathsf{MaxClique} $ can not be approximated within a $n^{1-o(n)}$ factor of the optimal \cite{engebretsen2000clique}. 
The best known upper bound for the approximation ratio is by Feige \cite{feige2004approximating} of $O(n (\log \log n / \log^3 n) )$ which improves upon  Boppana and Halldorsson's \cite{boppana1992approximating} result of $O(n / \log^2 n)$. 
Note that if there is a $f(n)$-approximation for $ \mathsf{MaxClique} $ then whenever the clique number of the graph is $\omega(f(n))$, the approximation algorithm returns a non-trivial clique (not a singleton vertex).

\begin{theorem}
	Given a graph $G$ whose maximum degree is $\Delta$, in polynomial time we can construct another graph $G_{clique}$, with $|G_{clique}| = O(n(\Delta+1)^k)$, such that if the clique number of $G_{clique}$ is $\kappa$ then $mr(G, \pi, k) = \kappa$.
\end{theorem}	

\noindent In the above theorem the graph $G_{clique}$ will be an $n$-partite graph. 
Hence $\kappa \le n = O(|G_{clique}|/(\Delta+1)^k)$.
As long as we have $(\Delta+1)^k = O(\log^2 n)$ we can use the approximation algorithm for $ \mathsf{MaxClique} $  to get a non-trivial approximation ratio of $O(n \log \log n / \log n)$.
	
\begin{proof}
	Here we give the reduction from $ \mathsf{MaxRoute} $ to $ \mathsf{MaxClique} $.
		First we augment $G$ by adding self-loops. Let this new graph be $G'$.
		Hence we can make every matching in $G'$ perfect by assuming each unmatched vertex is matched to itself.
	    Observe that any routing scheme on $G'$ induces a collection of walks for each pebble.
		This collection of walks are constrained as follows. 
		Let $W_i$ and $W_j$ corresponds to walks of pebbles starting at vertices $i$ and $j$ respectively. 
		Let $W_i[t]$ be the position of the pebble at time step $t$.
		They must satisfy the following two conditions:
		1)  $W_i[t] \neq W_j[t]$ for all $t \ge 0$.
		2)  $W_i[t+1] = W_j[t]$ iff $W_i[t] = W_j[t+1]$.
		Now consider two arbitrary walks in $G'$. 
		We call them compatible iff they satisfy the above two conditions. 
		We can check if two walks are compatible in linear time.
		
		Let $\mathcal{W}_i$ be the collection of all possible length $k$ walks starting from $i$ and ending at $\pi(i)$.
		Note that $|\mathcal{W}_i| = O((\Delta+1)^k)$.
		 For each $w \in \mathcal{W}_i$ we create a vertex in $G_{clique}$. Two vertices $u,v$ in $G_{clique}$ are adjacent if they do not come from the same collection ($u \in \mathcal{W}_i$ then $v \not \in \mathcal{W}_i$)  and $u$ and $v$ are compatible walks in $G'$. 
		Clearly, $G_{clique}$ is $n$-partite, where each collection of vertices from $\mathcal{W}_i$ forming a block.
		Furthermore, if $G_{clique}$ has a clique of size $\kappa$ then
		it must be the case that there are $\kappa$ mutually compatible walks in $G'$. 
		These walk determines a routing scheme (since they are compatible) that routes $\kappa$ pebbles to their destination.
		Now if $G_{clique}$ has a clique number $< \kappa$ then the largest collection of mutually compatible length $k$ walks must be $< \kappa$. Hence number of pebbles that can be routed to their destination in at most $k$ steps will be $< \kappa$. 
		
		In order to get a non-trivial approximation ratio we require that $(\Delta+1)^k = O(\log^2 n)$ which implies that the above reduction is polynomial in $n$. This completes the proof.

\end{proof}

%%%%%%%%%%%%%%%%%%%%%%%%%%%%%%%%%%%%%%%%%%%%%%%%%

\section{Structural Results On The Routing Number}

%%%%%%%%%%%%%%%%%%%%%%%%%%%%%%%%%%%%%%%%%%%%%%%%%
\subsection{An Upper Bound For $h$-connected Graphs}
%%%%%%%%%%%%%%%%%%%%%%%%%%%%%%%%%%%%%%%%%%%%%%%%%
It was shown in \cite{alon1994routing} that if a graph $G$ is $h$-connected then its routing number has a lower bound of $\Omega(n/h)$. This is easy to see since there exists $h$-connected graphs which have a balanced bipartition with respect to some cut-set of size $h$. For such a graph the permutation that routes every pebble from one partition to the other and vice-versa takes at least $\Omega(n/h)$ matchings. In this section we give an upper bound. Let $G_h$ be a induced connected subgraph of $G$ having $h$ vertices, we will show $rt(G) = O(n\ rt(G_h)/ h)$. Hence if $G$ has a $h$-clique then $rt(G) = O(n/h)$. In fact the result is more general. If $G_h$ is an induced subgraph with $\le h$ vertices such that $r = rt(G_h)/|G_h|$ is minimized then $rt(G) = O(nr)$.

We use the classical Lovasz-Gyori partition theorem for $h$-connected graphs for this purpose:

\begin{theorem}[Lovasz-Gyori]
	If $G$ is a $h$-connected graph then for any choice of positive numbers $n_1,\ldots,n_h$  with $n_1+\ldots+n_h=n$ and any set of vertices $v_1,\ldots,v_h$ there is a partition of the vertices $V_1,\ldots,V_h$ with $v_i \in V_i$ and $|V_i| = n_i$ such that the induced subgraph $G[V_i]$ is connected for all $1 \le i \le h$. 
\end{theorem}

We prove a combinatorial result.
We have $a$ lists $L_i$, $1\le i\le a$, each of length $b$.
Each element of a list is a number $c$, $1 \le c\le a$.
Further, across all lists, each number $c$ occurs exactly $b$ times.

\begin{lemma}%\ref{lemmax}
	Given lists as described, there exists an $a\times b$ array $A$ such that the
	$i$th row  is a permutation of $L_i$ and each column is a permutation of $\{1,2,3,\ldots,a\}$.
\end{lemma}

\begin{proof}
	By Hall's Theorem for systems of distinct representatives \cite{hall1998combinatorial}, we know that we can choose a representative from each $L_i$ to form the first column of $A$.
	The criterion of Hall's Theorem is that, for any $k$, any set of $k$ lists have at least $k$ distinct
	numbers; but there are only $b$ of each number so $k-1$ numbers can not fill up $k$ lists.
	Now remove the representative from each each list, 
	and iterate on the collection of lists of length $b-1$.
\end{proof}

To prove our upper bound we need an additional lemma.

\begin{lemma}%\ref{lemmay}
	Given a set $S$ of $k$ pebbles and tree $T$ with $k$ pebbles on its $k$ vertices. 
	Suppose we are allowed an operation that replaces the pebble at the root of $T$ by a pebble from $S$.
	We can replace all the pebbles in $T$ with the pebbles from $S$ in $\Theta(k)$ steps, each a replace or a matching step.
\end{lemma}

\begin{proof}
	Briefly, as each pebble comes from $S$ it is assigned a destination vertex in $T$, in reverse level order (the root is at level 0). 
	After a replace-root operation, there are two matching steps; these three will repeat.
	The first matching step uses disjoint edges to move elements of $S$ down to an odd level
	and the second matching step moves elements of $S$ down to an even level. 
	Each matching moves every pebble from $S$, that has not reached its destination, towards its destination.
	The new pebbles move without delay down their paths in this pipelined scheme.
	(The invariant is that each pebble from $S$ is either at its destination, or at an even level before the next replace-root operation.)
\end{proof}
\begin{theorem}
	If $G$ is $h$-connected and $G_h$ is an  induced connected subgraph of order $h$ then $rt(G) = O(n\ rt(G_h)/h)$.
\end{theorem}
\begin{proof}
	Let $V_h = \{u_1,\ldots,u_h\}$ be the vertices in $G_h$. We take these vertices as the set of $k$ vertices in Theorem 2. We call them ports as they will be used to route pebbles between different components. Without loss of generality we can assume $p = n/h$ is an integer. Let $n_1=n_2=\ldots=n_h=p$ and $V_i$ be the block of the partition such that $u_i \in V_i$. Let $H_i = G[V_i]$. Then for any permutation $\pi$ on $G$:
	\begin{enumerate}
		\item Route the pebbles in $H_i$ according to some permutation $\pi_i$. Since $H_i$ has $n/h$ vertices and is connected it takes $O(n/h)$ matchings.
		\item Next use $G_h$, $n/h$ times, to route pebbles between different partitions. We show that this can be done in $O(n\ rt(G_h)/h)$ matchings. (The ``replace-root'' step of Lemma 4, is actually the root replacements done by routing on $G_h$.)
		\item Finally, route the pebbles in each $H_i$ in parallel. Like step 1, this also can be accomplished in $O(n/h)$ matchings.
	\end{enumerate}
	Clearly the two most important thing to attend to in the above procedure are the permutations  in step 1 and the routing scheme of step 2.
	We can assume that each $H_i$ is a tree rooted at $u_i$ (since each $H_i$ has a spanning tree). Thus the decomposition looks like the one shown in Figure 5.

	\begin{figure}[h]
		\includegraphics[width=5cm]{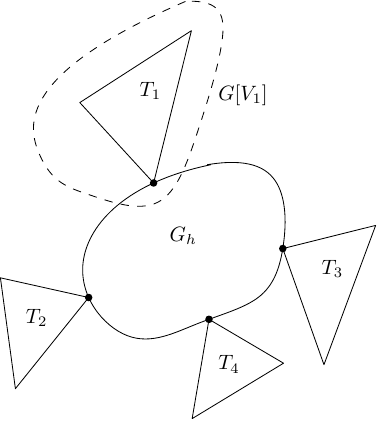}
		\centering
		\caption{$G$ is decomposed into 4 connected blocks, which are connected to each other via $G_h$.} 
	\end{figure}
	
	The permutation $\pi$ on $G$ indicates for each element of $H_i$, which $H_j$ it wants to be routed to, where $j$ could be $i$.
	So each $H_i$ can build a list $L_i$ of indices of the ports of $G_h$ that it wants to route its 
	elements to  (again, possibly to its own port).
	The lists satisfy the conditions of Lemma 3, with $a=h$ and $b=n/h$,
	We will use the columns of the array $A$ to specify the permutations routed using $G_h$ in step 2.
	Note that step 1 will need to preprocess each $H_i$ so that the algorithm of Lemma 4
will automatically deliver the elements of $H_i$ up to $u_i$ in the order specified by the $i$th row of $A$.
	
	Once the pebbles are rearranged in step 1, we use the graph $G_h$ to route them to their destination components. Each such routing takes $rt(G_h)$ steps. Between these routings on $G_h$ the incoming pebble at any of the port vertices is replaced by the next pebble to be ported; this requires 2 matching steps as seen in Lemma 4. Hence, after $rt(G_h) + 2$ steps a set of $h$ pebbles are routed to their destination components. This immediately gives the bound of the theorem. 
\end{proof}

%%%%%%%%%%%%%%%%%%%%%%%%%%%%%%%%%%%%%%%%%%%%%%%%%
\subsection{Relation Between Clique Number and Routing Number}
%%%%%%%%%%%%%%%%%%%%%%%%%%%%%%%%%%%%%%%%%%%%%%%%%
\begin{theorem}
	For a connected graph $G$ with clique number $\kappa$ its routing number is bounded by $O(n-\kappa)$. 
\end{theorem}

\begin{proof}
	Let $H$ be a clique in $G$ of size $\kappa$. Let $G_{\setminus H}$ be the minor of $G$ after the contraction of the subgraph $H$. Let the vertex that $H$ has been contracted to be $v$. Further, let $T$ be a spanning tree of $G_{\setminus H}$. When routing on $G_{\setminus H}$ we can treat $v$ as any other vertex of $G_{\setminus H}$. Taking into account the fact that $v$ can store more than one pebble internally. 
	When $v$ participates in a matching with some other vertex $u$ in $G_{\setminus H}$  we assume that exchanging pebbles takes 3 steps.
This accounts for the fact that the pebble thats need to be swapped with the pebbles at $u$ was not on a vertex adjacent to $u$ in the un-contracted graph $G$.
  The basic idea is to break the routing into two steps. 
  In the first step we simply move all pebbles in $v$ whose final detination is not in $v$ (i.e. not in un-contracted $H$) out.
  For a tree, it is known that \cite{benjamini2014acquaintance} we can route a subset of $p$ pebbles where each pebble needs to be moved at most $l$ distance in $\le p + 2l$ steps.
  Since $T$ has a diameter at most $n-\kappa$ and at most $min(\kappa, n-\kappa)$ pebbles need to be moved out of $v$ the first step can be accomplished $\le 3(n -\kappa) + O(1)$ steps. 
    At this point we can employ any optimal tree routing algorithm on $T$ where we charge 3 time units whenever $v$ is part of the matching to route all the pebbles in $G_{\setminus H}$.
  If we use the algorithm presented in \cite{zhang1999optimal} then we see that the routing takes at most $15/2(n-\kappa) + o(n)$ steps for any permutation.
  
\end{proof}

%\bibliography{mybibfile}

%
%-----------------------------------------------------------------------------

\end{document}